\documentclass[review]{elsarticle}
\usepackage[latin9]{luainputenc}
\usepackage{float}
\usepackage{amsthm,amsfonts}
\usepackage{amsmath}
\usepackage{amsthm}

\makeatletter

\floatstyle{ruled}
\newfloat{algorithm}{tbp}{loa}[section]
\providecommand{\algorithmname}{Algorithm}
\floatname{algorithm}{\protect\algorithmname}

\theoremstyle{plain}
\newtheorem{thm}{\protect\theoremname}[section]
\ifx\proof\undefined
\newenvironment{proof}[1][\protect\proofname]{\par
\normalfont\topsep6\p@\@plus6\p@\relax
\trivlist
\itemindent\parindent
\item[\hskip\labelsep\scshape #1]\ignorespaces
}{%
\endtrivlist\@endpefalse
}
\providecommand{\proofname}{Proof}
\fi
\theoremstyle{plain}
\newtheorem{lem}{\protect\lemmaname}[section]
\newtheorem{pro}[thm]{\protect\propositionname}
\newtheorem{cor}{\protect\corollaryname}[section]

\theoremstyle{definition}
\newtheorem{exam}{\protect\examplename}[section]

\usepackage{lineno}
\modulolinenumbers[5]

\makeatother

\providecommand{\theoremname}{Theorem}
\providecommand{\lemmaname}{Lemma}
\providecommand{\propositionname}{Proposition}
\providecommand{\corollaryname}{Corollary}
\providecommand{\examplename}{Example}

\begin{document}
\begin{frontmatter}

\title{Independent sets of closure operations}

\author[nhs]{Nguyen Hoang Son }

\ead{nhson@hueuni.edu.vn}

\address[nhs]{Department of Mathematics, College of Sciences, Hue University, Viet
Nam }

\begin{abstract} In this paper independent sets of closure operations are introduced. We characterize minimal keys and antikeys of closure operations in terms of independent sets. We establish an expression on the connection between minimal keys and antikeys of closure operations based on independent sets. We construct two combinatorial algorithms for finding all minimal keys and all antikeys of a given closure operation based on independent sets. We estimate the time complexity of these algorithms. Finally, we give an NP-complete problem concerning nonkeys of closure operations.
\end{abstract}
\end{frontmatter}

\section{Introduction}
Closure operations and closure systems appear in many fields in pure or applied mathematics and computer science. Many papers have appreared concerning lattices and combinatorial problems in closure operations which especially are the closed sets, minimal keys and antikeys of closure operations (see \cite{bu2,ca,de2,ng2,son}). The subsets $X\subseteq U$ satisfying $f(X)=U$ are called the keys of closure operation $f\in Cl(U)$. Clearly, this concept has an important role for the database. The data of a key determine the individual uniquely. The really important ones are the minimal keys: they are keys containing no other key as a proper subset (see e.g. \cite{bu2, de2, son}). The antikeys of closure operations (i.e. maximal non-keys) play an essential role in extremal problems of closure operations as well as in finding minimal keys (e.g., see \cite{bu2, son}). The closed set $Y$ of $f$ are defined by $f(Y)=Y$. The set of all closed sets of $f$ is called the closure system or meet-semilattice. Also, the closed sets  have been widely studied (see \cite{bu2, ca,ng2, son}).

The theory of hypergraphs is an important subfield of discrete mathematics with many relevant applications in both theoretical and applied computer science. Especially, it is a very useful tool for solution of combinatorial problems (e.g., see \cite{ber, de3,e}). The transversals and minimal transversals of a hypergraph are important concepts in this theory. The set of all minimal keys and the set of antikeys of closure operations form simple hypergraphs.

In this paper, we introduce the notion of independent sets of closure operation $f$. With a subset $X\subseteq U$, independent set $I(X)$ of $f$ be a subset determined by $U\setminus f(X)$. Denote by $I(f)$ and by $MI(f)$, the family of all independent sets and the family of all minimal independent sets of $f$, respectively. We show that generating all minimal keys of $f$ can be reduced to generating all minimal transversals of family $MI(f)$. We also give a representation of the set of all antikeys of $f$ in terms of independent sets. Based on these results, we establish the connection between minimal keys and antikeys of closure operations. Also, we construct two combinatorial algorithms finding all minimal keys and antikeys of a given closure operation by independent sets. Finally, in this paper we give an NP-complete problem. 

The paper is structured as follows. After an introduction section, in Section 2, we recall the definitions and the basic results of closure operations and hypergraphs. In Section 3 we introduce the notion of independent sets of closure operations. After that, we characterize the set of all minimal keys and the set of all antikeys of a closure operation by  independent sets. From these results, we establish an expression on the connection between minimal keys and antikeys of closure operations. In Section 4 we construct two combinatorial algorithms for finding all minimal keys and all antikeys of a given closure operation based on independent sets. We prove that the time complexitys of these algorithms are exponential in the number of elements of $U$. Finally, we give an NP-complete problem concerning nonkey of closure operations in Section 5.

\section{Definitions and preliminaries}
In this section, we recall the definitions and the preliminary results. The definitions and results in this section can be found in \cite{be,ber, da,de4,de1,son}.

Let $U$ be a nonempty finite set. Here $\mathcal{P}(U)$ denotes the power set of $U$, that is, the set of all subsets of $U$.  The mapping $f:\mathcal{P}(U)\to\mathcal{P}(U)$ is called a \textit{closure
operation} on $U$ if it satisfies the following conditions

(C1) $X\subseteq f(X)$,

(C2) $X\subseteq Y$ implies $f(X)\subseteq f(Y)$,

(C3) $f(f(X))=f(X)$,

\noindent for every $X,Y\subseteq U$.

We denote by $Cl(U)$ the set of all closure operations on $U$. 

Let $f\in Cl(U)$ and $X\subseteq U$. Set $X$ is called \textit{closed} of $f$ if $f(X)=X$. The family of closed sets is denoted $Closed(f)$. Therefore, $Closed(f)=\{X\subseteq U : f(X)=X\}$. It is easy to see that $U\in Closed(f)$ and $X,Y\in Closed(f)\Rightarrow X\cap Y\in Cloesd(f)$. Then we also can rewrite  $Closed(f)=\{f(X) : X\subseteq U\}$.

A family $\mathcal S$ of subsets of $U$ is called a \textit{closure system} (or \textit{meet-semilattice}) on $U$ if it satisfies the following conditions

(S1) $U\in \mathcal S$

(S2) $\forall \mathcal A\subseteq \mathcal P(U),  \emptyset \not=\mathcal A \subseteq \mathcal S \Rightarrow \bigcap \mathcal A\in \mathcal S$.

It can be seen that, if $\mathcal S$ is a closure system, and we define $f_{\mathcal S}(X)$ as
\[f_{\mathcal S}(X)=\bigcap \{Y\in \mathcal S:X\subseteq Y\}
\]
then $f_{\mathcal S}\in Cl(U)$. Conversely, if $f\in Cl(U)$, then there is exactly one closure system $\mathcal S$ on $U$ so that $f=f_{\mathcal S}$, where
\[\mathcal S=\{X\subseteq U:f(X)=X\}.
\]

Thus, $Cloesd(f)$ is a closure system. This means that there is a 1-1 correspondence between closure operations and closure systems. 

Let $f\in Cl(U)$ and $K\subseteq U$. Set $K$ is called a \textit{key} of $f$ if  $f(K)=U$. A key is called \textit{minimal} if every $a\in K$ then $f(K\setminus\{a\})$ not a key. We denote by $Key(f)$ the set of all minimal keys of $f$. 

A subset $K^{-1}\subseteq U$ is called a \textit{antikey} of $f$ if $f(K^{-1})\not=U$ and $\forall a\in U\setminus K^{-1},f(K^{-1}\cup\{a\})=U$. Denote $Antikey(f)$ the set of all antikeys of $f$. Therefore, we can see that $Antikey(f)$ is the set of all maximal nonkeys of $f$.

In \cite{son} we proved the connection between  minimal keys and antikeys of closure operations as follows:
$$\bigcup Key(f)=U\setminus \bigcap Antikey(f).$$

\begin{exam}\label{ex21}
The following mappings are basic closure operations:

(1) \textit{A maximal mapping} $m:\mathcal{P}(U)\to\mathcal{P}(U)$
is determined by $m(X)=U$ for every $X\subseteq U$. Then 
$$Closed(m)=\{U\}, Key(m)=\{\emptyset\} \text{ and } Antikey(m)=\emptyset.$$

(2) \textit{An identity mapping} $i:\mathcal{P}(U)\to\mathcal{P}(U)$
is determined by $i(X)=X$ for every $X\subseteq U$. Then 
$$Closed(i)=\mathcal P(U), Key(i)=\{U\} \text{ and } Antikey(i)=\{U\setminus \{a\}: a\in U\}.$$

(3) \textit{A translation mapping }$t_M:\mathcal{P}(U)\to\mathcal{P}(U)$
is determined by $t_M(X)= M\cup X$, where $M$ is a given
subset of $U$ and for every $X\subseteq U$. Then 
$$Closed(t_M)=\{M\cup X:X\subseteq U\}, Key(t_M)=\{U\setminus M\} \text{ and }$$ 
$$Antikey(t_M)=\{U\setminus \{a\}: a\in U\setminus M\}.$$
\end{exam}

Next, we introduce briefly the main concepts of hypergraph which will be needed in sequels. A hypergraph $\mathcal H$ is a pair $(V,\mathcal E)$, where $V$ is a finite set and $\mathcal E$  is a family of subsets of $V$. The elements of $V$ are called \textit{vertices}, and the elements of $\mathcal E$ \textit{edges}. Note that some authors, e.g. \cite{ber}, is required that the edge-set as well as each edge must be nonempty and that the union of all edges yields the vertex set. In this paper we do not require this. It is easy to see that a graph is a hypergraph with $|E|=2, \forall E\in \mathcal E$. For notational convenience, we will identify a hypergraph with its edge-set and vice versa if there is no danger of ambiguity. Therefore for hypergraph $\mathcal H=(V,\mathcal E)$, we write $E\in \mathcal H$ for $E\in \mathcal E$ etc.

A hypergraph $ {\cal   H}$ is called {\it simple} if it satisfies 
$$\forall E_i, E_j\in  {\cal   H}, E_i \subseteq E_j\Rightarrow E_i=E_j.$$ 

It can be seen that $Key(f)$ and $Antikey(f)$ are simple hypergraphs on $U$. 

A set $T\subseteq V$ is called a {\it transversal} of $ {\cal   H}$ (sometimes it is called {\it hitting set}) if it meets all edges of $ {\cal   H}$, i.e., $\forall E\in  {\cal   H}, T\cap E\not= \emptyset.$ We denote by $Trs({\cal  H})$ the family of all transversals of ${\cal   H}$. A transversal $T$ of $ {\cal   H}$ is called {\it minimal} if no proper subset $T'$ of $T$ is a transversal. The family of all minimal transversals of $ {\cal   H}$ called the transversal hypergraph of $ {\cal   H}$, and denoted by $Tr( {\cal   H})$. Clearly, $Tr( {\cal   H})$ is a simple hypergraph. 

Note that $\emptyset$ is a minimal transversal of $\mathcal H=\emptyset$, since each $T\subseteq V$ vacuously satisfies the transversal criterion, and that $\mathcal H=\{\emptyset\}$ has no transversal, since no $T\subseteq V$ has nonempty intersection with $\emptyset$.

The transversal hypergraphs have the following some basic properties

\begin{pro}[\cite{ ber}] \label{pro21} Let $\cal H$ and $\cal G$ two simple hypergraphs on $V$. Then

(1) $\mathcal H=Tr(\mathcal G)$ if and only if $\mathcal G=Tr(\mathcal H)$,

(2) $Tr(\mathcal H)=Tr(\mathcal G)$ if and only if $\mathcal H=\mathcal G$,

(3) $Tr(Tr(\mathcal H))=\mathcal H$.
\end{pro}

In \cite{de3} an algorithm finding the family of all minimal transversals of a given hypergraph (by induction) is presented.

\begin{algorithm}[H]
\caption{ (Finding the family of all minimal transversals)}
\label{al21}

Input: Let $ {\cal   H}=\{E_1, \ldots, E_m\}$ be a hypergraph on $V$.

Output: $Tr( {\cal   H}).$

Method:

Step 0. We set $\mathcal L_1=\{\{a\}: a\in E_1\}$. It is obvious that $\mathcal L_1=Tr(\{E_1\})$.

Step $q+1$. ($q<m$) Assume that
$$\mathcal L_q=\mathcal S_q\cup \{B_1, \ldots, B_{t_q}\},$$

\noindent where $B_i\cap E_{q+1}=\emptyset, i=1, \ldots, t_q $ and $\mathcal S_q=\{A\in \mathcal L_q : A\cap E_{q+1}\not= \emptyset\}$.

\smallskip
For each $i$ $(i=1, \ldots, t_q)$ constructs the set $\{B_i\cup \{b\} : b\in E_{q+1}\}$. Denote them by $A_1^i, \ldots, A_{r_i}^i (i=1, \ldots, t_q)$. Let
$$\mathcal L_{q+1}=\mathcal S_q\cup \{A_p^i  : A\in \mathcal S_q \Rightarrow A\not\subset A_p^i, 1\leq i\leq t_q, 1\leq p\leq r_i\}.$$

Let $Tr(\mathcal H)=\mathcal L_m$.
\end{algorithm} 
It can be seen that the determination of $Tr({\cal   H})$ based on Algorithm \ref{al21} does not depend on the order of $E_1, \ldots, E_m$. The time complexity of Algorithm \ref{al21} is exponential in $n$. However, in many cases, Algorithm \ref{al21} is very effective. Indeed, if we denote $\mathcal L_q=\mathcal S_q\cup \{B_1, \ldots, B_{t_q}\}$, $l_q =|\mathcal L_q|$ $(1\leq q \leq m-1)$  and $n=|V|$, it can be seen that the worst-case time complexity of Algorithm \ref {al21} is
$${\cal   O}(n^2\sum_{q=0}^{m-1}t_qu_q),$$
where $l_0=t_0=1$ and
$$
u_q=\begin{cases} l_q-t_q, & \text { if } l_q>t_q;\\
		1, & \text { if } l_q=t_q.
\end{cases}
$$

Clearly, in each step of Algorithm \ref{al21}, $\mathcal L_q$ is a simple hypergraph. It is known that the size of arbitrary simple hypergraph on $V$ cannot be greater than $\begin{pmatrix} 
		 n \\
		 \left\lfloor n/2 \right\rfloor 
\end{pmatrix}$, and
$$\begin{pmatrix} 
		 n \\
		 \left\lfloor n/2 \right\rfloor 
\end{pmatrix}\simeq \cfrac{2^{n+1/2}}{(\pi.n)^{1/2}}.$$ 

From this, the worst-case time complexity of Algorithm \ref{al21} cannot be more than exponential in the $n$. In cases for which $l_q\leq l_m$ $(q=1, \ldots, m-1)$, it is easy to see that the time complexity of Algorithm \ref{al21} is not greater than ${\cal   O}(n^2m|Tr({\cal   H})|^2).$ Thus, in these cases this algorithm finds $Tr({\cal   H})$ in polynomial time in $n,m$ and $|Tr({\cal   H})|$. Obviously, if $m$ is small, then this algorithm is very effective. It only requires polynomial time in $n$.

We now illustrate Algorithm \ref{al21} by the simple example as follows 
\begin{exam} Let $V=\{a,b,c,d,e\}$ and hypergraph $\mathcal H =\{\{a,c\},\{b,c,e\},$ $\{c,d\}\}$. Then we have

$\mathcal L_1=\{\{a\},\{c\}\}$

$\mathcal L_2=\{\{a,b\},\{a,e\},\{c\}\}$

$\mathcal L_3=\{\{a,b,d\},\{a,e,d\},\{c\}\}$.

Consequently 
$$Tr(\mathcal H)=\{\{a,b,d\},\{a,e,d\},\{c\}\}.$$
\end{exam}

\section{Independent set, minimal key and antikey of closure operations}
In this section, we introduce the concept of independent sets of closure operations. We characterize minimal keys and antikeys of closure operations in terms of independent sets.

Let $f\in Cl(U)$ and $X\subseteq U$. Set 
$$I(X)=\{a\in U: a\not\in f(X)\}.$$

Then we say that $I(X)$ is an \textit{independent set} of $f$. Denote by $I(f)$ the family of all independent set of $f$. Clearly, we have   $I(X)=U\setminus f(X)$ and $\emptyset \in I(f)$. Furthermore, we are easy to see that $X$ is a key of $f$ if and only if $I(X)=\emptyset$.

Note that $I(i)=Closed(i)=\mathcal P(U)$, and hence $|I(i)|=|Closed(i)|=2^{|U|}$. Consequently, we have $1\leq |I(f)|\leq 2^n$ for all $f\in Cl(U)$ and $|U|=n$.

We next set 
$$MI(f)=\{Y\in I(f): Y\not=\emptyset, (\forall Z\in I(f)\Rightarrow Y\not\subseteq Z)\}.$$

Family $MI(f)$ is called the family of all minimal independent sets of $f$. It can be seen that $MI(f)$ is a simple hypergraph on $U$. Remark that the number of elements of $MI(f)$ is always very small. We denote by $MAX(\mathcal S)$ the family of maximal elements of family $\mathcal S\subseteq \mathcal P(U)$. By the definition of the independent set of closure operation $f$, we can easily see that family $MI(f)$ can be represented by closure system $Closed(f)$ as follows:
$$\overline{MI(f)}=MAX(Closed(f)\setminus \{U\})$$
or
$$MI(f)=\overline{MAX(Closed(f)\setminus \{U\})}.$$

We consider again Example \ref{ex21}. Then we have 

$\bullet$ $I(m) = \{\emptyset\}, MI(m)= \emptyset$.

$\bullet$ $I(i) = \mathcal P(U),  MI(i)= \{\{a\}: a\in U\}$.

$\bullet$ $I(t_M) = \{U\setminus X\setminus  M: X\subseteq U\}, MI(t_M) = \{\{a\}:a\in U\setminus M\} $.

Now we study minimal keys of a closure operation by means of independent sets.  We show that generating all minimal keys of a closure operation $f$ can be reduced to generating all minimal transversals of a family of all minimal independent sets $MI(f)$.

First, we use the following helpful lemmas:

\begin{lem}\label{le31} If $Y\not=\emptyset$ is an independent set of $f$, then $U\setminus Y$ is not a key of $f$.
\end{lem}
\begin{proof} Let $\emptyset \not= Y=I(X)\in I(f)$. Suppose that $f(U\setminus Y)=U$. By the properties of closure operations and the definition of the independent set of closure operations, we have
$$f(X)=f(f(X))=U.$$

It follows that
$$U\setminus Y=U\setminus I(X)=f(X)=U.$$

Therefore, $Y=\emptyset$. This contradicts the hypothesis $Y\not=\emptyset$.
\end{proof}

\begin{lem}\label{le32}$X$ is not a key of $f$ if and only if $U\setminus X$ is a transversal of $Key(f)$.
\end{lem}
\begin{proof} Assume that $U\setminus X\not\in Trs(Key(f))$. This means that there exists a minimal key $K\in Key(f)$ such that $(U\setminus X)\cap K=\emptyset$, or $K\subseteq X$. Thus, $X$ is a key of $f$, which contradicts with the fact $f(X)\not=U$.

Conversely, suppose that $f(X)=U$. It follows that there is a minimal key $K\in Key(f)$ such that $K\subseteq X$. Thus, $(U\setminus X)\cap K=\emptyset$, or $U\setminus X\not\in Trs(Key(f))$. This contradicts with the hypothesis $U\setminus X$ is a transversal of $Key(f)$.
\end{proof}
\begin{thm}\label{th31} Let $f\in Cl(U)$. Then
$$Tr(Key(f))= MI(f).$$
\end{thm}
\begin{proof} Suppose that $T$ is a minimal transversal of $Key(f)$. By Lemma \ref{le32} we imply that $f(U\setminus T)\not= U$. Notice that $T\not=\emptyset$. Clearly, if $U\setminus T\subset f(U\setminus T)$ then for all $K\in Key(f), \ (U\setminus f(U\setminus T))\cap K \not= \emptyset.$ This contradicts $T\in Tr(Key(f))$. Therefore, we get $U\setminus T\in Closed(f)$. Then we have
$$I(U\setminus T)=U\setminus f(U\setminus T)=U\setminus (U\setminus T)=T.$$

This means that $T\in I(f)$. We now assume that there exists a $\emptyset \not=S\in I(f)$ so that $S\subset T$. According to Lemma \ref{le31}, $f(U\setminus S)\not= U$. Moreover, by Lemma \ref{le32} we obtain that $S$ is a transversal of $Key(f)$, which contradicts with  $T\in Tr(Key(f))$. Thus, $T\in MI(f)$ holds.

Conversely, suppose that $Y\in MI(f)$. Obviously, $Y\not=\emptyset$. Using Lemma \ref{le31}, we have $f(U\setminus Y)\not=U$. This means that for each $K\in Key(f)$, $Y\cap K\not= \emptyset$. Hence, $Y\in Trs(Key(f))$. We now assume that there is a $Z\in Tr(Key(f))$ such that $Z\subset Y$. Using the above proof, we also obtain $Z\in MI(f)$. This contradicts with the fact that $MI(f)$ is a simple hypergraph. Therefore, $Y\in Tr(Key(f))$ holds.
\end{proof}
From Theorem \ref{th31} and Proposition \ref{pro21}, we obtain a representation of the set of all minimal keys of closure operations as follows
\begin{cor}\label{co31}For every $f\in Cl(U)$, then
$$Key(f)=Tr(MI(f)).$$
\end{cor}
The antikeys of closure operations can also be characterized by independent sets. We know that \cite{son} antikeys of a closure operation $f$ have the following basic characterization 
$$Antikey(f)=MAX(Closed(f)\setminus \{U\}).$$

Using this characterization and the above observation of $MI(f)$, we also obtain the following a representation of the set of all antikeys of closure operations

\begin{cor}\label{co32} For every $f\in Cl(U)$, then
$$Antikey(f)=\overline{MI(f)}.$$
\end{cor} 
Combining Theorem \ref{th31}, Corollary \ref{co32} and Proposition \ref{pro21}, we establish the interesting connection between minimal keys and antikeys of closure operations as follows: 
\begin{thm}Let $f\in Cl(U)$. Then
$$Antikey(f)=\overline{Tr(Key(f))}.$$
\end{thm}
We consider again Example \ref{ex21}. Then we get  

$\bullet$ $Key(m)=Tr(MI(m))=Tr(\emptyset)=\{\emptyset\}$,

$Antikey(m)=\overline{MI(m)}=\emptyset$.

$\bullet$ $Key(i)=Tr(MI(i))=Tr(\{\{a\}:a\in U\})=\{U\}$,

$Antikey(i)=\overline{MI(i)}=\{U\setminus \{a\}:a\in U\}$.

$\bullet$ $Key(t_M)=Tr(MI(t_M))=Tr(\{\{a\}: a\in U\setminus M\})=\{U\setminus M\}$,

$Antikey(t_M)=\overline{MI(t_M)}=\{U\setminus \{a\}:a\in U\setminus M\}$.

\section{Generating the set of all minimal keys and antikeys of a closure operation based on independent sets}
Based on the results presented in Section 3, in this section, first we shall present a effective combinatorial algorithm finding all minimal keys of a given closure operation by independent sets. However, notice that finding a minimal key of a closure operation $f\in Cl(U)$ is efficiently possible: it is easy to see that  $U=\{a_1,a_2,\ldots,a_n\}$ is a key $f$. If we define $K_0=U$, and for all $i=1,2,\ldots,n$
$$K_i= \begin{cases} K_{i-1}\setminus \{a_i\}, & \text {if } f(K_{i-1}\setminus \{a_i\})=U;\\
		K_{i-1}, & \text { otherwise}
\end{cases} $$ 
then $K_n\in Key(f)$. 

Therefore, it can be seen that a minimal key of $f$ can be found in polynomial time in $n$. 

\begin{algorithm}[H]
\caption{(Finding all minimal keys based on independent sets)}
\label{al41}

Input: $f\in Cl(U)$ with $U=\{a_1,a_2,\ldots,a_n\}$.

Output: $Key(f)$.

Method:

Step 1: Construct the family of all independent sets of $f$:
$$I(f)=\{I(X): X\subseteq U\},$$
where $I(X)=\{a\in U: a\not\in f(X)\}.$

Step 2: From $I(f)$ we compute the family of all minimal independent sets of $f$:
$$MI(f)=\{Y\in I(f): Y\not=\emptyset, (\forall Z\in I(f)\Rightarrow Y\not\subseteq Z)\}.$$

Step 3: Using Algorithm \ref{al21} we compute $Tr(MI(f))$. Then let $Key(f)=Tr(MI(f))$.
\end{algorithm}
According to Corollary \ref{co31}, Algorithm \ref{al41} computes exactly $Key(f)$. Clearly, the number of elements of $MI(f)$ is very small, Algorithm \ref{al21} computing $Tr(MI(f))$ is very effective. Thus,  it can be seen that the time complexity of our algorithm is the time complexity of Step 1 and Step 2. This means that the time complexity of our algorithm is exponential in the $n$.

Next, we present a effective combinatorial algorithm finding all antikeys of a given closure operation by independent sets.

\begin{algorithm}[H]
\caption{(Finding all antikeys based on independent sets)}
\label{al42}

Input: $f\in Cl(U)$ with $U=\{a_1,a_2,\ldots,a_n\}$.

Output: $Antikey(f)$.

Method:

Step 1: Construct the family of all independent sets of $f$:
$$I(f)=\{I(X): X\subseteq U\},$$
where $I(X)=\{a\in U: a\not\in f(X)\}.$

Step 2: From $I(f)$ we compute the family of all minimal independent sets of $f$:
$$MI(f)=\{Y\in I(f): Y\not=\emptyset, (\forall Z\in I(f)\Rightarrow Y\not\subseteq Z)\}.$$

Step 3: From $MI(f)$ we construct $\overline{MI(f)}$. Then let $Antikey(f)=\overline{MI(f)}$.
\end{algorithm}
By Corollary \ref{co32}, it is easy to see that Algorithm \ref{al42} computes exactly $Antikey(f)$. Because the number of elements of $MI(f)$ is very small, Step 3 requires polynomial time in $n$. Thus, it can be seen that the time complexity of our algorithm is the time complexity of Step 1 and Step 2.  This means that the time complexity of our algorithm is exponential in the $n$.

The following example shows that for a given closure operation $f\in Cl(U)$, Algorithm \ref{al41} and Algorithm \ref{al42} can be applied to find all minimal keys and antikeys of $f$.
\begin{exam}\label{ex31} Let us consider the mapping $f_a:\mathcal{P}(U)\to\mathcal{P}(U)$, where $a\in U$, as follows:
$$f_a(X)= \begin{cases} U, & \text {if } a\in X;\\
		X, & \text { otherwise}
\end{cases} $$ 

It is easy to see that $f_a\in Cl(U)$. Then we have 
$$I(f_a)=\{U\setminus f_a(X): X\subseteq U\}$$
and thus 
$$MI(f)=\{\{a\}\}.$$

It implies that 
$$\overline{MI(f)}=\{U\setminus \{a\}\}$$
$$Tr(MI(f))=\{\{a\}\}.$$

Consequently, the set of all minimal keys and the set of all antikeys of $f$ are
$$Key(f)=\{\{a\}\}$$
$$Antikey(f)=\{U\setminus \{a\}\}.$$
\end{exam}

\section{Complexity of  the problem of nonkeys of closure operations}
Let $f\in Cl(U)$. Denote $\mathcal S=\{X: X \text{ is not a key of } f\}$. Obviously, $\mathcal S$ is the set of all nonkeys of $f$. From the definition of antikeys, we can see that $Antikey(f)$ is the set of all maximal nonkeys of $f$. It is known that maximal nonkeys play important roles for extremal problems of closure operations as well as for many other problems. This section we give an  NP-complete problem concerning nonkey of closure operations. The problem we will consider can be described as follows:

\textbf{Name:} NONKEY OF CLOSURE OPERATION (NONKEY)

\textbf{Instance:} A closure operation $f\in Cl(U)$ and an integer $k$ such that $k\leq |U|$. 

\textbf{Question:} Is there a nonkey $X$ such that $k\leq |X|$.

In order to show the NP-completeness of NONKEY, we will use the independent set problem which is known to be NP-complete \cite{ga}:

\textbf{Name:} INDEPENDENT SET  (IS)

\textbf{Instance:}  A integer $k$ and an undirected graph $G=(V,E)$, where $V$ is the set of vertices and $E$ is the set of edges.

\textbf{Question:} Is there an independent set $I$ having cardinality greater than or equal to $k$.
\begin{thm}\label{th51} The NONKEY problem is NP-complete.
\end{thm} 
\begin{proof} We nondeterministically choose a set $X$ such that $k\leq |X|$ and decide whether $X$ is a nonkey of $f$. Obviously, by the definition of closure operation $f\in Cl(U)$, our algorithm is nondeterministic polynomial. Therefore, our problem lies in NP.

Now, we shall prove that IS problem is polynomially reducible to our problem. Indeed, let $G=(V,E)$ be an undirected graph with $k\leq |V|$. We define a map $f: \mathcal P(U)\to \mathcal P(U)$ as follows: 
$$f(X)= \begin{cases} V, & \text {if } X=\{u,v\}\in E\\
		X, & \text {otherwise}
\end{cases} $$ 
where $U=V$. Note that for each $X\not\in E$ we have $f(X)=i(X)$. It is easy to see that $f\in Cl(U)$, and $f$ is constructed in polynomial  time in the size $G$.

According to the definition of the set of edges, it is clear that $E$ is a simple hypergraph on $V$. Since $E$ is the set of edges, and by the definition of the minimal key of closure operation, we can obtain that $Key(f)=E$. Hence, $X$ is not a key of $f$ if and only if $\{u,v\}\not\subseteq X$ for all $\{u,v\}\in E$. Consequently, $X$ is nonkey of $f$ if and only if $X$ is an independent set of $G$.
\end{proof}
 Because $Antikey(f)$ is the set of all maximal nonkeys of closure operation $f$, according to Theorem \ref{th51} we can see that if $NP\not= P$, then there is no polynomial time algorithm finding $Antikey(f)$ from a given closure operation $f$.

\section{Conclusions}
We studied independent sets of closure operations and characterized minimal keys and antikeys of closure operations in terms of independent sets. We established  an expression the connection between minimal keys and antikeys of closure operations based on independent sets. We also constructed two combinatorial algorithms for finding all minimal keys and all antikeys of a given closure operation based on independent sets. We showed that time complexitys of these algorithms are exponential in $n$. Finally, we proved that for a given closure operation $f$ and an integer $k$, the problem deciding whether there exists a nonkey of $f$ having cardinality greater than or equal to $k$ is NP-complete.


\section*{References}
\begin{thebibliography}{}

\bibitem {be}F. E. Bennett, Li Sheng Wu, On minimum matrix representation of closure operations, Discrete Applied Mathematics, 26 (1990) 25--40.

\bibitem{ber} C. Berge, Hypergraphs: combinatorics of finite sets, North-Holland, Amsterdam, 1989.

\bibitem{bu2}G. Burosch, J. Demetrovics,  G.O. H. Katona, D. J. Kleitman, A. A. Sapozhenko, On the number of closure operations, Combinatorics, Paul Erdos is Eighty, 1 (1993) 91--105. 

\bibitem{ca} N. Caspard, B. Monjardet, The lattices of closure systems, closure operators, and implicational systems on a finite set: a survey, Discrete Applied Mathematics, 127 (2003) 241--269.

\bibitem {da}V. Danilov, G. Koshevoy, Choice functions and extensive operators, Order, 26 (2009) 69--94.

\bibitem{de4} J. Demetrovics, G. Gyepesi, A note on minimal matrix representation of closure operations, Combinatorica,  3 (1983) 177--179.

\bibitem {de1} J. Demetrovics,  Z. Furedi,  G. O. H. Katona, Minimum  matrix representation of closure operations, Discrete Applied Mathematics, 11 (1985) 115--128.

\bibitem{de2}  J. Demetrovics, G. Hencsey , L. Libkin , I. Muchnik, On the interaction between closure operations and choice functions
with applications to relational databases, Acta Cybernetica, 10 (1992) 129--139.

\bibitem{de3} J. Demetrovics, Vu Duc Thi, Describing candidate keys by hypergraphs, Computers and Artificial Intelligence, 18 (1999) 191--207.

\bibitem {e}T. Eiter, G. Gottlob, Identifying the minimal transversals of a hypergraph and related problems, SIAM Journal on Computing, 24 (1995) 1278--1304.

\bibitem{ga} M. R. Garey, D. S. Johnson, Computers and Intractability: A Guide to the Theory of NP-Completeness, A Series of Books in the Mathematical Sciences, W. H. Freeman and Company, San Francisco, 1979.

\bibitem{ru} S. Rudolph, Succinctness and tractability of closure operator representations, Theoretical Computer Science,  658 (2017) 327--345.

\bibitem {ng2}Vu Duc Nghia,  Relationships between closure operations and choice functions equivalent descriptions of a family of functional
dependencies, Acta Cybernetica, 16 (2004) 485--506. 

\bibitem{son} Nguyen Hoang Son, Vu Duc Thi, Some the combinatorial characteristics of closure operations, Algebra and Discrete Mathematics, 28 (2019) 144-156.
\end{thebibliography}
\end{document}